\documentclass[11pt, twoside, letter]{article}
\usepackage[margin = 1in]{geometry}
\usepackage{comment}
\usepackage{indentfirst}
\usepackage[pdftex]{graphicx}
\usepackage{amsmath, amsthm}
\usepackage{amssymb}
\usepackage{framed}
\usepackage{aliascnt}
\usepackage{hyperref}
\usepackage{multirow, multicol}
\usepackage{subfigure, subfloat, graphicx} 
\usepackage{verbatim, rotating, paralist}
\usepackage{caption}
\usepackage{pstricks, sgamevar, egameps}

\newcommand{\Real}{\mbox{${\mathbb R}$}}

\newcommand{\xx}{\mbox{\boldmath $x$}}
\newcommand{\yy}{\mbox{\boldmath $y$}}
\newcommand{\vv}{\mbox{\boldmath $v$}}
\newcommand{\rr}{\textbf{r}}
\newcommand{\qq}{\mbox{\boldmath $q$}}

\newcommand{\zz}{\mbox{\boldmath $z$}}

\newcommand{\pp}{\mbox{\boldmath $p$}}
\newcommand{\ww}{\mbox{\boldmath $w$}}
\newcommand{\D}{\mbox{$\Delta$}}
\newcommand{\ones}{\mbox{${\bf 1}$}}
\newcommand{\zeros}{\mbox{${\bf 0}$}}
\newcommand{\defeq}{\stackrel{\textup{def}}{=}}

\newcommand{\CL}{{\mathcal{L}}}

\newcommand{\ssp}{SSP}

\newtheorem{theorem}{Theorem}[section]
\newaliascnt{lemma}{theorem}
\newtheorem{lemma}[lemma]{Lemma}
\newtheorem{definition}[theorem]{Definition}
\newtheorem{claim}[theorem]{Claim}
\newtheorem{fact}[theorem]{Fact}
\newtheorem{remark}[theorem]{Remark}

\aliascntresetthe{lemma}

\newcommand{\etal}{{\em et al.~}}

\def\Pr{\ensuremath{\mathrm{Pr}}}

\begin{document}

\title{The Complexity of Genetic Diversity
}  

\author{
Ruta Mehta
\\Georgia Institute of Technology\\ rmehta@cc.gatech.edu
\and Ioannis Panageas\\Georgia Institute of Technology\\ioannis@gatech.edu
\and Georgios Piliouras\\Singapore University of Technology and Design\\georgios@sutd.edu.sg
\and Sadra Yazdanbod\\Georgia Institute of Technology\\yazdanbod@gatech.edu }

\date{}

\maketitle

\thispagestyle{empty}

\begin{abstract}
A key question in biological systems is whether genetic diversity persists in the long run under evolutionary competition or whether a single  dominant genotype emerges. Classic work by Kalmus in 1945 \cite{Kalmus45} has established that even in simple diploid species (species with two chromosomes) diversity can be guaranteed as long as the heterozygote individuals enjoy a selective advantage. 
Despite the classic nature of the problem, as we move towards increasingly polymorphic traits (\textit{e.g.}, human blood types) predicting diversity and understanding its implications is still not fully understood. Our key contribution is to establish complexity theoretic hardness results implying that even in the textbook case of single locus diploid models predicting whether diversity survives or not given its fitness landscape is algorithmically intractable. 

We complement our results by establishing that under randomly chosen fitness landscapes diversity survives with significant probability. Our results are structurally robust along several dimensions (e.g., choice of parameter distribution, different definitions of stability/persistence, restriction to typical subclasses of fitness landscapes). Technically, our results exploit connections between game theory, nonlinear dynamical systems, complexity theory and biology and establish hardness results for predicting the evolution of a deterministic variant of the well known multiplicative weights update algorithm in symmetric coordination games which could be of independent interest.

\end{abstract}

\newpage
\clearpage
\setcounter{page}{1}

\section{Introduction}

The beauty and complexity of natural ecosystems have always been a source of fascination and inspiration for the human mind.
The exquisite biodiversity of Galapagos' ecosystem, in fact, inspired Darwin to propose his theory of natural selection as an explanatory 
mechanism for the origin and evolution of species. This revolutionary idea can be encapsulated in the catch phrase "survival of the fittest"\footnote{The phrase "survival of the fittest" was coined by Herbert Spencer.}.
 Natural selection promotes the survival of those genetic traits that provide to their carriers an evolutionary advantage.
 
A marker of the true genius behind any paradigm is its longevity and in this respect the success of natural selection is without precedent.
The formalized study of evolution, dating back to the work of Fisher, Haldane, and Wright in the beginning of the twentieth century, still and to 
a large extent focuses on simple, almost toy-like, alas concrete, models of this famous aphorism and the experimental and theoretical analysis of them. 
Despite this unassuming facade, if one digs a little deeper, and delves even into the most basic of questions, how/under what circumstances does natural selection  
 preserve genetic diversity, a puzzling web 
  of connections is revealed that couples biology, game theory, nonlinear dynamical systems, topology, and computationally complexity 
 into a tightly interwoven bundle. Our understanding of this web of relationships is still rather fragmented with the classic results in the area focusing 
 on isolated facets of this space.
 
Arguably, the most influential result in the area of mathematical biology is Fisher's fundamental theorem of natural selection (1930)~\cite{Fisher30}.
It states that the rate of increase in fitness of any organism at any time is equal to its genetic variance in fitness at that time.
In the classical model of population genetics (Fisher-Wright-Haldane, discrete or continuous version) of single loci multi-allele diploid models it implies that
the average fitness of the species populations is always strictly increasing unless we are at an equilibrium. In fact, convergence to equilibrium is pointwise
even if there exist continuum of equilibria (See \cite{akin} and references therein). From a dynamical systems perspective, this a rather strong characterization, since it establishes that the average fitness acts as
a Lyapunov function for the system and that every trajectory converges to an equilibrium. 

Besides the purely dynamical systems interpretation,  an alternative, more palpable, game theoretic interpretation of 
these genetic systems is possible. Specifically, these systems can be interpreted as symmetric coordination/partnership two-agent games\footnote{A coordination/partnership game is a game where at each outcome all agents receive the same utility.} where each agent applies a (discrete-time version of) replicator dynamics\footnote{Replicator dynamics (as well as their discrete variants) are close analogues to the well known multiplicative updates (MWUA) family of dynamics \cite{Kleinberg09multiplicativeupdates}.} and where both agents start from the same mixed initial strategy.  
The analogies are as follows: The proportions of different alleles in the initial population encode the starting common mixed strategy. The random mating in the next time epoch is captured as two symmetric agents applying the same mixed strategy (allele mixture) against each other. Each pairing of two alleles/strategies defines an individual with a specific fitness. This is the common utility of each of the two agents at the specific outcome. Therefore, the game between the two agents is a coordination/partnership game and furthermore since both Aa and aA allele combinations encode the same individual the payoff/fitness matrix is also symmetric. Finally, the updates of the mixed strategies/allele frequencies promote the survival of the fittest strategies as captured exactly by replicator/MWUA dynamics.

 In game theoretic language, the fundamental theorem of natural selection implies that the social welfare of the game acts as potential for the game dynamics. The dynamics converge to mixed strategies such that the expected utility of each strategy played with positive probability are equal to each other, a superset of the symmetric Nash equilibria.

These topological, game theoretic results do not provide any insight on the survival of genetic diversity. This depends on the polymorphism as well as the stability 
of the equilibria. One way to formalize this question is to ask whether (given a specific fitness landscape, symmetric coordination game) there exists an equilibrium with support size of at least two which has a region of attraction of positive measure. The answer to this question for the minimal case of two alleles (alleles $a/A$, individuals $aa/aA/AA$) is textbook knowledge and can be traced back to the classic work of Kalmus (1945)~\cite{Kalmus45}. The intuitive answer here is that diversity can survive when the heterozygote individuals, $aA$, have a fitness advantage. Intuitively, this can be explained by the fact that even if evolution tries to dominate the genetic landscape by $aA$ individuals, the random genetic mixing during reproduction will always produce some $aa,AA$ individuals, so the equilibrium that this process is bound to reach will be mixed. On the other hand, it is trivial to create instances where homozygote individuals are the dominant species regardless of the initial condition. 

As we increase the size/complexity of the fitness landscape, not only is not clear that a tight characterization of the diversity-inducing fitness landscape exists (a question about global stability of nonlinear dynamical systems), but moreover, it is even less clear whether one can decide efficiently whether such conditions are satisfied by a given fitness landscape (a computational complexity consideration).  How can one address this combined challenge and furthermore, how can one account for the apparent genetic diversity of the ecosystems around us? 

In a nutshell, we establish that the decision version of the problem is computationally hard. This core result is shown to be robust across a number of directions. Deciding the existence of stable polymorphic equilibria remains hard under a host of different definitions of stability examined in the dynamical systems literature. The hardness results persist even if we restrict the set of allowable landscape instances to reflect typical instance characteristics. Despite the hardness of the decision problems, randomly chosen fitness landscapes are shown to support polymorphism with significant probability (at least $1/3$). The game theoretic interpretation of our results allow for proving hardness results for understanding standard game theoretic dynamics  
in symmetric coordination games. We believe that this is an important result of independent interest as it points out at a different source of complexity in understanding social dynamics. Even in economic settings, where the computational complexity of identifying a sample (Nash) equilibrium is trivial (as in symmetric two-agent coordination games) the intricacy of social dynamics may add a layer of complexity that obfuscates even the most basic characteristics of future system states.

\section{Technical Overview}

To study survival of diversity in diploidy, we need to understand characteristics of the limiting population under evolutionary pressure. 
We focus ob the simplest case of single locus (gene) species. For this case, evolution under natural selection has been shown to follow replicator dynamics in symmetric two-player coordination games~\cite{akin}, where
the genes on two chromosomes are players and alleles are their strategies. Losert and Akin established pointwise convergence  through a potential/Lyapunov function argument \cite{akin}; here average fitness is the potential.
The limiting population corresponds (almost surely) to stable fixed points (FP) of this dynamic. 
Thus to make predictions about
diversity we need to characterize and compute these limiting fixed-points. 

Let $\CL$ denote the set of fixed points (FP) with region of attraction of positive measure. 
 Given a random starting point our dynamics converge to such a FP with positive probability.
As it turns out an exact characterization of $\CL$ may be
unlikely.  Instead we try to capture it as closely as possible through different stability notions. First we consider two standard
notions defined by Lyapunov, the stable and asymptotically stable FP. Stable fixed-point are fixed points such that if we start close enough to them we
stay close for all time, 
while in case of asymptotically stable FP, if we start close enough the trajectory furthermore converges to it.
(Definition  \ref{def.afp}). Thus asymptotically stable $\subseteq \CL$ follows.


Characterization of these stability notions using the absolute values of eigenvalues (EVal) of the Jacobian is well known: if the Jacobian at a fixed-point has an Eval $>1$ then the FP is {\em un-}stable (complement of stable), and if all Eval $<1$  then it is asymptotically stable. The case when all Eval $\le 1$ with equality holding for some, is the ambiguous one. In that case we can say nothing about the stability because the Jacobian does not suffice, \textit{i.e.}, we need to compute the Hessian etc. We will call these fixed points {\em linearly-stable}.

Under some assumptions on the update rule of the dynamics (the function has to be a diffeomorphism) it is true that at a fixed-point, say $\xx$, if some EVal $>1$ then the direction of corresponding eigenvector is repelling, and therefore any starting vector with a component of this vector can never converge to $\xx$. Thus points converging to $\xx$ can not have positive measure. Using this as an intuition we show that $\CL \subseteq$ {\em linearly-stable} FPs. In other words points converging to linearly-{\em un}-stable fixed-points have zero measure
(Theorem \ref{generic_zero}). This theorem is heavily utilized to understand (non-)existence of diversity.

Efficient verification requires efficient computation. However, note that whether a fixed-point is (asymptotically) stable or not does not seem easy to verify. To achieve this we define two more notions: {\em Nash stable} and {\em strict Nash stable}. 

It is easy to see that Nash equilibria of the corresponding coordination game are fixed-points of the replicator dynamics (but not vice-versa).
Keeping this in mind we define {\em Nash stable} FP, which is a Nash equilibrium and the sub-matrix corresponding to its support
satisfies certain $-ve$ semi-definiteness. The latter condition is derived from the fact that stability relates to local optima and also from Sylvester's law of inertia \cite{wikiintertia}. 
For {\em strict Nash stable} both conditions are strict, namely strict Nash equilibrium and $-ve$ definite. Combining all of these
notions we show the following:

\begin{theorem}[Informal]
$\ $

\vspace{-0.7cm}
\[
\begin{array}{lcl}
\mbox{Strict Nash stable} & \subseteq & \mbox{Asymptotically stable}\  \subseteq \ \CL \ \subseteq \ \mbox{linearly-stable} \ = \
\mbox{Nash stable}\\
& &\ \ \ \ \ \ \  \rotatebox{90}{$\supseteq$}  \\
& &\ \ \  \ \ \mbox{stable} \subseteq \ \mbox{linearly-stable} \ = \ \mbox{Nash stable}\\
\end{array}
\]
\end{theorem}

We give an example where these notions don't coincide. Let $x_{t+1}$ be the next step for the following update rules: $f(x_t)=\frac{1}{2}x_t, g(x_t) = x_t - \frac{1}{2}x_t^2, h(x_t) = x_t + \frac{x_t^3}{2}, d(x_t)=x_t$. Then $0$ is asymptotically stable, stable and linearly stable and $0 \in \CL$ for $f$, it is linearly stable and $\in \CL$ but is not stable/as. stable for $g$, it is linearly stable but not stable/as. stable and $\notin \CL$ for $h$ and finally it linearly stable and stable but not as.stable and $\notin \CL$ for $d$.
   
We must stress the fact that generically, these sets coincide. It can be shown \cite{lyubich} for example that given your fitness matrix, if you give a random perturbation, then all fixed-points will have Jacobian with eigenvalues that don't lie on the unit circle. These fixed-points are called hyperbolic. Formally, if you consider the dynamics as an operator (called \textit{Fisher} operator) then the set of hyperbolic operators is dense in the space of Fisher operators. 
 
Our primary goal was to see if diversity is going to survive. We formalize this by checking 
 whether set $\CL$ contains a mixed point, i.e., where
more than one strategies have non-zero probability, implying that 
{\em diversity survives with some positive probability}.  In Section \ref{sec.hard} we show that for all five notions of stability, checking
existence of mixed FP is NP-hard. This in turn gives NP-hardness for checking survival of diversity as well. 

\begin{theorem}[Informal]
Given a symmetric matrix $A$, it is NP-hard to check if it has mixed (asymptotically) stable, linear-stable, or (strict) Nash
stable fixed-point.
This implies that it is NP-hard to check whether diversity survives for a given fitness matrix.
\end{theorem}

\begin{figure}[hbtp]
\vspace{-0.3cm}
\begin{minipage}{0.695\textwidth}
All our reductions are from {\bf $k$-clique} - given an undirected graph check if it has a clique of size $k$; a well known NP-hard
problem.  We design one main reduction and show that it works for all the notions of stability using Theorem \ref{thm.char}. 
Here, as shown in Figure \ref{fig1}, the main idea is to use a modified version of adjacency matrix $E$ of the graph as one of the
blocks in the payoff matrix, such that, $(i)$ clique of size $k$ or more implies a stable Nash equilibrium in that block,
and $(ii)$ all stable mixed equilibria are only in that block. 
Here $E'$ is modification of $E$ where off-diagonal zeros are replaced with $-h$ where $h$ is a large (polynomial-size number). 

\end{minipage}
\hspace{0.025\textwidth}
\begin{minipage}[c]{0.28\textwidth}
   \centering
\includegraphics[width=\linewidth,height=0.8\linewidth]{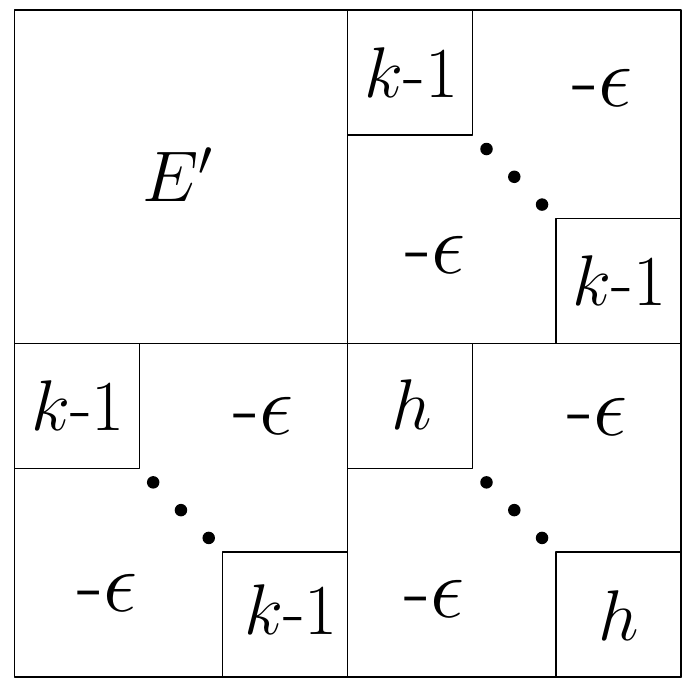}
\caption{Matrix of equations \ref{eq.red1}, section 6.1} 
\label{fig1}
\end{minipage}
\end{figure}

Let $A$ be $2n \times 2n$ symmetric fitness matrix of Figure \ref{fig1}, where first $n$ strategies correspond to $n$ node of the
graph. For $(i)$ given a $k$-clique we construct maximal clique containing it, say of size $m$, and then show that probability
$\frac{1}{m}$ for strategies corresponding to nodes in the clique, and zero otherwise, gives (strict Nash) stable FP.  To show $(ii)$
we use $-ve$ semi-definiteness property of Nash stable (which contains all other notions), which implies $A_{ii}\le 2A_{ij}$ for all $i, j$
in the support a Nash stable FP. Therefore if $h>2(k-1)$, then none of $\{n+1,\dots,2n\}$ can be in the support. Next we concentrate on
only strategies that have big enough probability, call it set $S$. Nash equilibrium property of Nash stable, $-h$s and $k-1$s entries in $A$ force each non-zero probability to be $\le \frac{1}{k}$. Using these facts we show $|S|\geq k$ and the corresponding vertices in $G$ form a clique. So We prove existence of $k$-clique.

The fitness matrices created in above hardness results are very specific, while one could argue that real-life fitness matrices may be
more random. So is it easy to check survival of diversity in a typical matrix? Or is it easy to check if a given allele survives? We
answer negatively for both of these.  There has been a lot of work on understanding decision versions of Nash equilibrium in general
games \cite{GZ,CS,SS,GMVY}, whereas to the best of our knowledge these are the first NP-hardness results for coordination games, and
therefore may be of independent interest.

Finally we show that even though checking is hard, on an average things are not that bad. In Section \ref{asec.survive} we prove that,

\begin{theorem}[Informal]
If the fitness matrix is picked at uniformly at random then with significantly high
probability, at least $\frac{1}{3}$, diversity will surely survive. 
\end{theorem}

Sure survival happens if every fixed-point in $\CL$ is mixed. We show that this fact is ensured if every diagonal entry $(i,i)$ of the fitness matrix is dominated by some entry in it's row or column. Essentially, this condition makes the corresponding two player coordination game to have no symmetric pure Nash equilibrium. By considering matrices with entries i.i.d from a continuous distribution (so that the probability to get the same number twice is zero), we show that with probability bounded away from zero for all size of matrices this condition is satisfied (theorem \ref{lem.pp}). To do this, we use some inclusion-exclusion argument to avoid the correlations (the matrix is symmetric so there are dependencies).

\section{Related Work}\label{sec.rel}



In the last  few years we have witnessed a rapid cascade of theoretical results on the intersection of computer science
and evolution. 
Livnat \etal \cite{PNAS1:Livnat16122008} introduced the notion of mixability, the ability of
an allele to combine itself successfully with other alleles within a specific population.
In
\cite{ITCS:DBLP:dblp_conf/innovations/ChastainLPV13,PNAS2:Chastain16062014}  connections where established between haploid evolution and
game theoretic dynamics in coordination games.
Even more recently Meir and Parkes~\cite{Meir15} provided a more detailed examination of these connections.
These dynamics are close variants of the standard (discrete) replicator dynamics~\cite{Hofbauer98}. 
Replicator dynamics is closely connected to the multiplicative weights update algorithm~\cite{Kleinberg09multiplicativeupdates}.
Analogous game theoretic interpretations are
known for diploids~\cite{akin}.
It is also possible to introduce connections between satisfiability and evolution~\cite{evolfocs14}.  

The error threshold is the rate of errors in genetic mixing above which genetic information disappears~\cite{Eig93}.
Vishnoi~\cite{Vishnoi:2013:MER:2422436.2422445} showed existence of such sharp thresholds. Moreover, in \cite{VishnoiSpeed}
he sheds light on the speed of evolution. 
Finally, in~\cite{DSV} Dixit, Srivastava and Vishnoi present finite
population models for asexual haploid evolution that closely track the standard infinite population model of 
Eigen~\cite{Eigen71}. Due to space limitations we present information related to biology in appendix, see~\ref{termsinbio}.

The complexity of computing whether a game has an evolutionarily stable
strategy has been studied first by Nissan and then by Etessami and Lochbihler~\cite{Nisan06,etessami2008computational}
and has recently been nailed down to be $\Sigma_2^P$-complete by Conitzer~\cite{Conitzer13}. 
These results are not directly comparable to our own since they apply for general
symmetric bimatrix games. More critically, these decision problems are completely orthogonal 
to understanding the persistence of genetic diversity. 

In~\cite{ITCS15MPP} Mehta, Panageas and Piliouras
examine the question of diversity for haploid species. 
Despite the systems' superficial similarities the two analyses come to starkly different conclusions.
In haploids systems all mixed (polymorphic) equilibria are unstable and
evolution converges to monomorphic states. In the case of diploid systems the answer to whether diversity survives or not depends crucially on the geometry of the fitness landscape. This is not merely a question about stability of fixed points but requires careful combinations of tools from topology of dynamical systems, stability analysis of fixed points as well as complexity theoretic techniques. 

\section{Preliminaries}\label{sec.prel}

In this section we briefly describe the evolutionary dynamics on diploids with single locus \cite{Nagylaki1,Nagylaki2,akin,lyubich},
and relate it to a $2$-player symmetric coordination game.\\ \noindent{\bf Notations:} We use boldface letters, like $\xx$, to denote
column vectors, and $x_i$ is its $i^{th}$ coordinate. Vectors $\zeros_n$ and $\ones_n$ are $n$-D vectors with all zeros and ones.  For
matrix $A$, $A_{ij}$ is the entry in $i^{th}$ row and $j^{th}$ column. By norm $||.||$ we mean $||.||_\infty$. $[n]$ denotes set
$\{1,\dots,n\}$, and $\Delta_n$ denotes $n$-D simplex, and let $SP(\xx)$ be $\{i\ |\ x_i >0\}$.

\subsection{Evolutionary dynamics}\label{sec.prelevol}
Consider a diploid single locus species, with $n$ alleles numbered $1,\dots,n$. The fitness of such a species can be represented by an
$n\times n$ matrix $A$, where $A_{ij}$ denotes fitness of an individual with chromosome pair $(i,j)$. Clearly, $A$ is a symmetric
matrix. Let $\xx\in \Delta_n$ represent the proportion of alleles in the current population, then under the evolutionary process of
natural-selection this proportion changes as per the following multi-variate function $f:\Delta_n \rightarrow \Delta_n$; in the
literature this is often called Discrete Replicator Dynamics.

\begin{equation}\label{eq.f}
\mbox{$\ \ \xx'=f(\xx)\ \ $ where }\ \ x'_i = f_i(\xx) = x_i \frac{(A\xx)_i}{\xx^TA\xx},\ \ \ \forall i \in [n]
\end{equation}

$f$ is a continuous function with convex, compact domain (= range), and therefore always has a fixed-point
\cite{brower}. Further, limit points of $f$ have to be fixed-points, i.e., $\xx$ such that $f(\xx)=\xx$. 

\begin{fact}\label{lem.fpchar}
Profile $\xx$ is a fixed-point of $f$ iff $\forall i\in [n],\ x_i> 0 \Rightarrow (A\xx)_i = \xx^TA\xx$. 
\end{fact}

A fundamental question is: starting from arbitrary population of alleles how does the population look like in limit under
evolutionary pressures? Does it even converge to a limit point? In other words is $f$ going to converge? If so, then how do these limit
points look like?
Fisher's Fundamental Theorem of Natural Selection \cite{akin,lyubich} says that mean fitness function 
$$\pi (\xx) =\xx^TA\xx$$ 
(potential function in game theory terms) satisfies the inequality $\pi(f(\xx)) \geq \pi(\xx)$ and the equality holds iff $\xx$ is a
fixed-point. Losert and Akin \cite{akin,lyubich} showed that the dynamics above converge point-wise to fixed-points and that the rule
of the dynamics $f$ is a diffeomorphism in $\Delta_{n}$.  



\subsection{Games, Nash equilibria and symmetries}\label{sec.prelgame}

In this paper we consider two-player (locus) games, where each player has finitely many pure strategies (alleles).  Let $S_i,\ i=1,2$
be the set of strategies for player $i$, and let $m\defeq|S_1|$ and $n\defeq|S_2|$, then such a game can be represented by two payoff
matrices $A$ and $B$ of dimension $m\times n$.  Players may randomize amongst their strategies.  The set of mixed strategies are
$\Delta_m$ and $\Delta_n$ respectively.

\begin{definition}{(Nash Equilibrium \cite{agt.bimatrix})}
A strategy profile is said to be a Nash equilibrium (NE) if no player can achieve a better payoff by a
unilateral deviation \cite{nash}. Formally, $(\xx,\yy) \in \D_m\times \D_n$ is a NE iff $\forall \xx' \in \D_m,\
\xx^TA\yy \geq \xx'^TA\yy$ and $\forall \yy' \in \D_n,\  \xx^TB\yy \geq \xx^TB\yy'$. 
\end{definition}

Game $(A,B)$ is said to be symmetric if $B=A^T$. In a symmetric game the strategy sets of
both the players are identical, i.e., $m=n$, and $S_1=S_2$. We will use $n$, $S$ and $\D_n$ to denote the strategy related quantities.
Strategy $\xx \in \D_n$ is a symmetric NE, i.e., both play $\xx$, if and only if

\begin{equation}\label{eq.sne}
\forall i \in S, x_i>0 \Rightarrow (A\xx)_i = \max_k (A\xx)_k
\end{equation}

\begin{definition}\label{def.ssne}
NE $\xx$ is {\bf strict} if $\forall k \notin SP(\xx),\ (A\xx)_k < (A\xx)_i$, where $i
\in SP(\xx)$.
\end{definition}

\noindent{\bf Symmetric Coordination Game.} In a coordination/partnership game $B=A$, i.e., both the players get the same payoff always. Thus if
$A$ is symmetric then $(A,A)$ is a symmetric coordination game.  The next lemma follows using (\ref{eq.sne}) and Fact \ref{lem.fpchar}

\begin{lemma}\label{lem.sne2fp}
If $\xx$ is a symmetric NE of game $(A,A)$, it is a fixed-point of dynamics (\ref{eq.f}).
\end{lemma}

From now on, by {\em mixed} (fixed-point) strategy we mean strictly mixed (fixed-point) strategy, i.e., $\xx$ such that
$|SP(\xx)|>1$, and non-mixed are called {\em pure}.

\subsection{Basics in Dynamical Systems}

\begin{definition}\label{def.sfp}
A fixed-point $\rr$ of $f : \Delta \to \Delta$ is called {\bf stable} if, for every $\epsilon > 0$, there exists a $\delta =
\delta(\epsilon) > 0$ such that, for all $\pp \in \Delta$ with $||\pp-\rr\|| < \delta$ we have that $||f^n(\pp)-\rr|| <
\epsilon$ for every $n \geq 0$, otherwise it is called {\bf unstable}.  \end{definition}

\begin{definition}\label{def.afp} 
A fixed-point $\rr$ of $f : \Delta \to \Delta$ is called {\bf asymptotically stable} if it stable and there
exists a (neighborhood) $\delta > 0$ such that, for all $\pp \in \Delta$ with $||\pp-\rr|| < \delta$ we have that
$||f^n(\pp)-\rr|| \to 0$ as $n \to \infty$.  \end{definition}

By definition it follows that if $\xx$ is asymptotically stable w.r.t. dynamics $f$ (\ref{eq.f}), then the set of
initial conditions in $\Delta$ so that the dynamics converge to $\xx$ has positive measure and hence $\xx \in \CL$, therefore $\textrm{ asymptotically stable} \subseteq \CL$.
Using the fact that under $f$ the potential function $\pi(\xx)=\xx^TA\xx$ strictly decreases unless $\xx$
is a fixed-point, the next theorem was derived in \cite{lyubich}.

\begin{theorem}\label{maximum} (\cite{lyubich}, $\S$ 9.4.7)  A fixed-point $\rr$ of dynamics (\ref{eq.f}) is stable if and only
if it is a local maximum of $\pi$, and is asymptotically stable if and only if it is a strict
local maximum.  \end{theorem}

As the domain of $\pi$ is closed and bounded, there exists a global maximum of $\pi$ in $\Delta_n$, which by
Theorem \ref{maximum} is a stable fixed-point, and therefore its existence follows. 
However, existence of asymptotically stable fixed-point is not guaranteed, for example if $A=[1]_{m\times n}$ then no
$\xx\in \Delta_n$ is attracting under $f$.

\subsubsection{Stability and Eigenvalues}
To analyze limiting points of $f$ with respect to the notion of stability we need to 
use the eigenvalues of the Jacobian of $f$ at fixed-points. Let $J^\rr$ denote the Jacobian at $\rr \in \Delta_n$. 
Following is a well-known theorem in dynamics/control theory that relates (asymptotically)
stable fixed-points with eigenvalue of its Jacobian. 

\begin{theorem}\label{thm.sfpJ} \cite{perko} At fixed-point $\xx$ if $J^{\xx}$ has at least one eigenvalue with absolute value $>1$,
then $\xx$ is unstable. If all the eigenvalues have absolute value $<1$  then it is asymptoticaly stable.
\end{theorem}

\begin{definition}\label{def.ls} A fixed-point $\textbf{r}$ is called { linearly stable}, if the eigenvalues $J^\rr$ are at most $1$
in absolute value. Otherwise, it is called {linearly unstable}.
\end{definition}


Theorem \ref{thm.sfpJ} implies that eigenvalues of the Jacobian at a stable fixed-point have absolute value at most $1$,
however the converse may not hold. Using properties of $J^{\xx}$ and \cite{akin}, we prove next:
(see appendix \ref{asec.char} for Jacobian equations). 

\begin{theorem}\label{generic_zero} The set of initial conditions in $\Delta_n$ so that the dynamics \ref{eq.f} converge to linearly
unstable fixed-points has measure zero.  
\end{theorem}
In Theorem \ref{generic_zero} we manage to discard only those fixed-points whose Jacobian has eigenvalue with absolute value $>1$, while characterizing limiting points of $f$; the latter is finally used to argue about the survival of diversity. Thus it is true that $\CL \subseteq\textrm{ linearly stable}$.

\section{Convergence, Stability, and Characterization}\label{sec.char}

Our primary goal is to capture limiting points of dynamics $f$ (\ref{eq.f}), where positive measure of starting points converge. We
denote this set by $\CL$. In this section we try to characterize $\CL$ using various notions of stability, which have game theoretic
and combinatorial interpretation.  These notions sandwich setwise the classic notions of stability given in the preliminaries,  give us
partial characterization of $\CL$, and are crucial for our hardness results and our results on the survival
of diversity. 

\subsection{Stability and Nash equilibria}

Given a symmetric matrix $A$, a two-player game $(A,A)$ forms a symmetric coordination game. 
In this section we identify special symmetric NE of this game to characterize stable fixed-points of $f$.
Given a profile $\xx \in \Delta_n$, define a transformed matrix $T(A,\xx)$ of dimension $(k-1)\times(k-1)$, where
$k=|SP(\xx)|$, as follows. 

{\small
\begin{equation*}\label{eq.t}
\mbox{Let $SP(\xx)=\{i_1,\dots,i_k\}$, $B=T(A,\xx)$.}\ \forall a,b < k,\ B_{ab}=A_{i_ai_b} + A_{i_ki_k} - A_{i_ai_k}-A_{i_ki_b}
\end{equation*}
}
\noindent Since $A$ is symmetric it is easy to check that $B$ is also symmetric, and therefore has real eigenvalues.
Recall the Definition \ref{def.ssne} of strict symmetric NE.

\begin{definition}\label{def.ns}
A strategy $\xx$ is called (strict) Nash stable if it is a (strict) symmetric NE of
the game $(A,A)$, and $T(A,\xx)$ is negative (definite) semi-definite.
\end{definition}

\begin{lemma}\label{lem.t}
For any given $\xx \in \Delta_n$, $T(A,\xx)$ is negative (definite) semi-definite iff ($\yy^TA\yy < 0$) $\yy^TA\yy \leq 0$,
$\forall y \in \Real^n$ such that $\sum_i y_i=0$ and $x_i=0\Rightarrow y_i=0$.
\end{lemma}
Using the fact that stable fixed-point are local optima, we map them to
Nash stable strategies.

\begin{lemma}\label{lem.sfp2ns} Every stable fixed-point $\rr$ of $f$ is a Nash stable of game $(A,A)$. 
\end{lemma}

Since stable fixed-points always exist, so do Nash stable strategies (Lemma \ref{lem.sfp2ns}). 
Next we map strict Nash stable strategies map to asymptotically stable fixed-points, as the negative definiteness and
strict symmetric Nash of the former implies strict local optima, and the next lemma follows. 

\begin{lemma}\label{lem.sns2afp} (\cite{lyubich} $\S$ 9.2.5) Every strict Nash stable is asymptotically stable.  \end{lemma}

The above two lemmas show: strict Nash stable $\subseteq$ asymptotically stable (by definition) $\subseteq$ stable
 (by definition) $\subseteq$ Nash stable. Further, by Theorem \ref{thm.sfpJ} and the definition of {\em
linearly stable} fixed-points we know that stable $\subseteq$ linearly-stable. What remains is the relation between 
Nash stable and linearly stable. The next lemma answers this.

\begin{lemma}\label{lem.nseqls} Strategy $\rr$ is Nash stable iff it is a linearly stable fixed-point.
\end{lemma}

Using Theorem \ref{maximum} and Lemmas \ref{lem.sfp2ns}, \ref{lem.sns2afp} and \ref{lem.nseqls} we get the following characterization among all the notions of stability that we have discussed so far. We also remind you that \\$\textrm{asymptotically stable }\subseteq \CL \subseteq \textrm{ linearly stable}$.  

\begin{theorem}\label{thm.char}
Given a symmetric matrix $A$, we have 

$$\begin{array}{l}
\textrm{strict Nash stable} \subseteq \textrm{asymptotically stable} (= \textrm{strict local optima})\\
\hspace{3cm}\subseteq \textrm{stable } (= \textrm{local optima}) \subseteq   \textrm{Nash
stable} = \textrm{linearly stable}\end{array}$$ 
\end{theorem}

As stated before, generically (random fitness matrix) we have hyperbolic fixed-points and all the previous notions coincide. Next we show invariance of (strict) Nash stable sets under affine transformation of $A$.

%

\begin{lemma}\label{lem.inv}
Let $A$ be a symmetric matrix, and $B=A+c$ for a $c\in \Real$, then the set of (strict) Nash stable strategies of $B$ are identical to that of $A$.
\end{lemma}

\section{NP-Hardness Results}\label{sec.hard}
Positive chance of survival of phenotypic (allele) diversity in the limit under the evolutionary pressure of selection (dynamics \ref{eq.f}),
implies existence of a mixed linearly stable fixed-point (Theorem \ref{generic_zero}). This notion encompasses all the
other notions of stability (Theorem \ref{thm.char}), and may contain points that are not attracting. Whereas, strict Nash stable and
asymptotically stable are attracting.

In this section we show that checking if there exists a mixed stable profile, for any of the five notions of stability (Definitions
\ref{def.sfp}, \ref{def.afp}, \ref{def.ls} and \ref{def.ns}), may not be easy.  In particular, we show that the problem of checking if
there exists a mixed profile that satisfy any of the stability condition is NP-hard. The reduction also gives NP-hardness for checking
if a given pure strategy is played with non-zero probability (subset) at these.  In other words, it is NP-hard to check if a particular
allele is going to survive in the limit under the evolution.  Finally we extend all the results to the typical class of matrices, where
exactly one diagonal entry is dominating (see Remark \ref{rem}). All the reductions are from {\bf k-Clique}, a well known NP-complete
problem \cite{CSLR}.

\begin{definition}
\textbf{(k-Clique)} Given an undirected graph $G=(V,E)$, with $V$ vertices and $E$ edges, and an integer $0< k < |V|-1=n-1$,
decide if $G$ has a clique of size $k$. 
\end{definition}	

\noindent{\bf Properties of $G$.} Given a simple graph $G=(V,E)$ if we create a new graph
$G'$ by adding a vertex $u$ and connecting it to all the vertices $v \in V$, then it is easy to see that graph $G$ has a
clique of size $k$ if and only if $G'$ has a clique of size $k+1$.  Therefore, wlog we can assume that there exists a vertex
in $G$ which is connected to all the other vertices. Further, if $n=|V|$, then such a vertex is the $n^{th}$ vertex. By
abuse of notation we will use $E$ an adjacency matrix of $G$ too, $E_{ij}=1$ if edge $(i,j)$ present in $G$ else it is zero. 

\subsection{Hardness for checking stability}\label{sec.SHard}
In this section we show NP-hardness (completeness for some) results for decision versions on (strict) Nash stable strategies and
(asymptotically) stable fixed-points. 
Given graph $G=(V,E)$ and integer $k<n$, we construct the following symmetric $2n \times 2n$ matrix $A$, where $E'$ is modification of $E$ where off-diagonal zeros are replaced with $-h$ where $h>2n^2+5$ .

{\small
\begin{equation}\label{eq.red1}
\forall i\le j,\  A_{ij}=A_{ji} = \left\{
  \begin{array}{ll}
    E'_{ij} & \mbox{ if } i,j \le n\\
k-1 & \mbox{ if } |i - j| = n \\ 
	h & \mbox{ if } i,j>n \mbox{ and } i=j, \mbox{ where $h > 2n^2+5$} \\
	-\epsilon & \mbox{otherwise, } \mbox{where $0<\epsilon\leq \frac{1}{10n^3}$} 
  \end{array} \right.
  \end{equation}
}

\noindent Clearly, $A$ is a symmetric but is not non-negative. Next lemma maps $k$-clique to mixed-strategy that is also strict Nash
stable fixed-point (FP). Note that such a FP satisfies all other stability notions as well, and hence implies existence of mixed limit
point in $\CL$.

\begin{lemma}\label{lem.cl2sns}
If there exists a clique of size at least $k$ in graph $G$, then the game $(A,A)$ has a mixed strategy $\pp$ that is strict Nash stable. 
\end{lemma}

Next we want to show the converse for all notions of stability, namely if mixed-strategy exists for any notion of the five notions of
stability then there is a clique of size at least $k$ in the graph $G$. Since each of the five stability implies Nash stability, it
suffices to map mixed Nash stable strategy to clique of size $k$. 
For this, and reductions that follow, we will use the following property due to negative semi-definiteness of Nash stability extensively.

\begin{lemma} \label{lem.dom} 
Given a fixed-point $\xx$, if $T(A,\xx)$ is negative semi-definite, then $\forall i \in SP(\xx), A_{ii}\le 2A_{ij},\ \forall j\neq i
\in SP(\xx)$. Moreover if $\xx$ is a mixed Nash stable then it has in its support at most one strategy $t$ with $A_{tt}$ is dominating. 
\end{lemma}

Nash stable also implies symmetric Nash equilibrium. Next lemma maps (special) symmetric NE to $k$-clique. 

\begin{lemma}\label{lem.sne2cl}
Let $\pp$ be a symmetric NE of game $(A,A)$. If $SP(\pp) \subset [n]$ and $|SP(\pp)|>1$, then there exists a
clique of size $k$ in graph $G$.
\end{lemma}

We obtain the next lemma essentially using Lemmas \ref{lem.dom} and \ref{lem.sne2cl}.

\begin{lemma}\label{lem.ns2cl}
If game $(A,A)$ has a mixed Nash stable strategy, then graph $G$ has a clique of size $k$.
\end{lemma}

The next theorem follows using Theorem \ref{thm.char}, Lemmas \ref{lem.cl2sns} and \ref{lem.ns2cl}, and the
property observed in Lemma \ref{lem.inv}. Since there is no polynomial-time checkable condition for (asymptotically) stable
fixed-points\footnote{These are same as (strict) local optima of function $\pi(\xx)=\xx^TA\xx$, and checking if a given 
$\pp$ is a local optima can be inconclusive if hessian at $\pp$ is (negative) semi-definite.}
its containment in NP is not clear, while for (strict) Nash stable strategies containment in NP follows from the Definition 
\ref{def.ns}. 

\begin{theorem}\label{thm.mainHard}
Given a symmetric matrix $A$, 
checking if (i) game $(A,A)$ has a mixed (strict) Nash stable (or linearly stable) strategy is NP-complete. (ii) dynamics $f$
(\ref{eq.f}) has a mixed (asymptotically) stable fixed-point is NP-hard. Even if $A$ is a positive matrix.
\end{theorem}

Note that since adding a constant to $A$ does not change its strict Nash stable and Nash stable
strategies \ref{lem.inv}, and since these two sandwiches all other stability notions, the second part of the above theorem follows. 

As noted in Remark \ref{rem}, matrix with iid entries from any continuous distribution has in expectation exactly one row with
dominating diagonal. One could ask does the problem become easier for this typical case. We answer negatively by extending all the
NP-hardness results to this case as well, where matrix $A$ has exactly one row whose diagonal entry dominates all other entries of the
row. See Appendix \ref{asec.ddhard} for details, and thus the next theorem follows.

\begin{theorem}\label{thm.1ns}
Given a symmetric matrix $A$, checking if
(i) game $(A,A)$ has a mixed (strict) Nash stable (or linearly stable) strategy is NP-complete.
(ii) dynamics (\ref{eq.f}) applied on $A$ has a mixed (asymptotically) stable fixed-point is NP-hard.
Even if $A$ is strictly positive, or has exactly one row with dominating diagonal.
\end{theorem}

\noindent{\bf Hardness for Subset.} Another natural question is whether a particular allele is going to survive with positive
probability in the limit for a given fitness matrix. In Appendix \ref{asec.subset} we show that this may not be easy either,
by proving hardness for checking if there exists a stable strategy $\pp$ such that $i\in SP(\pp)$ for a given $i$. In
general, given a subset $S$ of pure strategies it is hard to check if $\exists$ a stable profile $\pp$ with $S\subseteq
SP(\pp)$.

\medskip

\noindent{\bf Survival of Diversity and Hardness.} As discussed in Section \ref{sec.char} 
checking if diversity survives in the limiting population of single locus diploid organism reduces to checking
``if $f$ converges to a {\em mixed} fixed-point with positive probability''.  In absence of clear characterization of the {\em mixed}
limit points of $f$ in terms of any of the stability definition, the hardness does not follow directly from the above result. Appendix
\ref{asec.divNhard} explains how above results can be combined to obtain the following theorem. Also see Remark \ref{rem.cg} on
complexity of decision problem for general Nash equilibrium in coordination games.

\begin{theorem}\label{thm.main}
Given a fitness matrix $A$ for a diploid organism with single locus, it is NP-hard to decide if, under evolution, diversity will
survive (by converging to a specific mixed equilibrium with positive probability) when starting allele frequencies are picked iid from uniform distribution. Also, deciding if a given
allele will survive is NP-hard.  
\end{theorem}

Finally we show that even though checking is hard, on an average things are not that bad. Specifically we show that if every entry of a
fitness matrix is i.i.d. from a continuous distribution, then the probability that diversity will survive surely is pretty high. 
We refer the readers to Appendix \ref{asec.survive} for details.

\begin{theorem}\label{ithm3}
If the fitness matrix is picked at uniformly at random then with significantly high
probability, at least $\frac{1}{3}$, diversity will surely survive. 
\end{theorem}

\section{Acknowledgments}

\noindent
We would like to thank Prasad Tetali for helpful discussions and suggestions.

\bibliographystyle{plain}
\bibliography{sigproc5}

\begin{thebibliography}{10}

\bibitem{PNAS2:Chastain16062014}
Erick Chastain, Adi Livnat, Christos Papadimitriou, and Umesh Vazirani.
\newblock Algorithms, games, and evolution.
\newblock {\em PNAS}, 2014.

\bibitem{ITCS:DBLP:dblp_conf/innovations/ChastainLPV13}
Erick Chastain, Adi Livnat, Christos~H. Papadimitriou, and Umesh~V. Vazirani.
\newblock Multiplicative updates in coordination games and the theory of
  evolution.
\newblock In {\em ITCS}, volume abs/1208.3160, pages 57--58, 2013.

\bibitem{Conitzer13}
V.~Conitzer.
\newblock The exact computational complexity of evolutionarily stable
  strategies.
\newblock In {\em The 9th Conference on Web and Internet Economics (WINE)},
  2013.

\bibitem{CS}
Vincent Conitzer and Tuomas Sandholm.
\newblock New complexity results about {N}ash equilibria.
\newblock {\em Games and Economic Behavior}, 63(2):621--641, 2008.

\bibitem{CSLR}
Thomas~H. Cormen, Clifford Stein, Ronald~L. Rivest, and Charles~E. Leiserson.
\newblock {\em Introduction to Algorithms}.
\newblock McGraw-Hill Higher Education, 2nd edition, 2001.

\bibitem{DGP}
C.~Daskalakis, P.~W. Goldberg, and C.~H. Papadimitriou.
\newblock The complexity of computing a {N}ash equilibrium.
\newblock {\em SIAM Journal on Computing}, 39(1):195--259, 2009.

\bibitem{DSV}
Narendra~M. Dixit, Piyush Srivastava, and Nisheeth~K. Vishnoi.
\newblock A finite population model of molecular evolution: Theory and
  computation.
\newblock {\em J. Computational Biology}.

\bibitem{Eigen71}
M.~Eigen.
\newblock Self-organization of matter and the evolution of biological
  macromolecules.
\newblock {\em Die Naturwissenschaften}, 58:456--523, 1971.

\bibitem{Eig93}
M.~Eigen.
\newblock The origin of genetic information: Viruses as models.
\newblock {\em Gene}, 1993.

\bibitem{etessami2008computational}
Kousha Etessami and Andreas Lochbihler.
\newblock The computational complexity of evolutionarily stable strategies.
\newblock {\em International Journal of Game Theory}, 37(1):93--113, 2008.

\bibitem{Fisher30}
R.A. Fisher.
\newblock {\em The Genetical Theory of Natural Selection}.
\newblock Clarendon Press, Oxford, 1930.

\bibitem{GMVY}
J.~Garg, R.~Mehta, V.~V. Vazirani, and S.~Yazdanbod.
\newblock Etr-completeness for decision versions of multi-player (symmetric)
  nash equilibria.
\newblock 2015.

\bibitem{GZ}
I.~Gilboa and E.~Zemel.
\newblock {N}ash and correlated equilibria: Some complexity considerations.
\newblock {\em Games Econ. Behav.}, 1:80--93, 1989.

\bibitem{Hofbauer98}
J.~Hofbauer and K.~Sigmund.
\newblock {\em Evolutionary Games and Population Dynamics}.
\newblock Cambridge University Press, Cambridge, 1998.

\bibitem{brower}
V.I. Istratescu.
\newblock {\em Fixed Point Theory: An Introduction}.
\newblock Mathematics and Its Applications. Springer Netherlands, 2001.

\bibitem{Kalmus45}
H.~Kalmus.
\newblock Adaptive and selective responses of a population of drosophila
  melanogaster containing e and e+ to differences in temperature, humidity, and
  to selection for development speed.
\newblock {\em Journal of Genetics}, 47:58�63, 1945.

\bibitem{Kleinberg09multiplicativeupdates}
Robert Kleinberg, Georgios Piliouras, and {\'E}va Tardos.
\newblock Multiplicative updates outperform generic no-regret learning in
  congestion games.
\newblock In {\em STOC}, 2009.

\bibitem{PNAS1:Livnat16122008}
Adi Livnat, Christos Papadimitriou, Jonathan Dushoff, and Marcus~W. Feldman.
\newblock A mixability theory for the role of sex in evolution.
\newblock {\em PNAS}, 2008.

\bibitem{evolfocs14}
Adi Livnat, Christos Papadimitriou, Aviad Rubinstein, Andrew Wan, and Gregory
  Valiant.
\newblock Satisfiability and evolution.
\newblock In {\em FOCS}, 2014.

\bibitem{akin}
V.~Losert and E.~Akin.
\newblock Dynamics of games and genes: Discrete versus continuous time.
\newblock {\em Journal of Mathematical Biology}, 1983.

\bibitem{lyubich}
Yuri Lyubich.
\newblock {\em Mathematical Structures in Population Genetics}.
\newblock Springer-Verlag.

\bibitem{ITCS15MPP}
Ruta Mehta, Ioannis Panageas, and Georgios Piliouras.
\newblock Natural selection as an inhibitor of genetic diversity.
\newblock In {\em ITCS}, 2015.

\bibitem{Meir15}
Reshef Meir and David Parkes.
\newblock A note on sex, evolution, and the multiplicative updates algorithm.
\newblock In {\em Proceedings of the 12th International Joint Conference on
  Autonomous Agents and Multiagent Systems (AAMAS ?15)}, 2015.

\bibitem{Nagylaki1}
T~Nagylaki.
\newblock The evolution of multilocus systems under weak selection.
\newblock {\em Genetics}, 134(2):627--47, 1993.

\bibitem{Nagylaki2}
T.~Nagylaki, J.~Hofbauer, and P.~Brunovsk\'{y}.
\newblock Convergence of multilocus systems under weak epistasis or weak
  selection.
\newblock {\em J. Math Biol.}, (2):103--33, 1999.

\bibitem{nash}
J.~Nash.
\newblock Equilibrium points in n-person games.
\newblock {\em PNAS}, pages 48--49, 1950.

\bibitem{Nisan06}
N.~Nisan.
\newblock A note on the computational hardness of evolutionary stable
  strategies.
\newblock {\em Electronic Colloquium on Computational Complexity (ECCC)},
  13(076), 2006.

\bibitem{Soda15a}
I.~Panageas and G.~Piliouras.
\newblock http://arxiv.org/pdf/1403.3885v4.pdf, 2014.

\bibitem{perko}
Lawrence Perko.
\newblock {\em Differential Equations and Dynamical Systems}.
\newblock Springer, 1991.

\bibitem{SS}
Marcus Schaefer and Daniel \v{S}tefankovi\v{c}.
\newblock Fixed points, {N}ash equilibria, and the existential theory of the
  reals.
\newblock Manuscript, 2011.

\bibitem{shub}
Michael Shub.
\newblock {\em Global Stability of Dynamical Systems}.
\newblock Springer-Verlag, 1987.

\bibitem{Vishnoi:2013:MER:2422436.2422445}
Nisheeth~K. Vishnoi.
\newblock Making evolution rigorous: The error threshold.
\newblock ITCS '13, pages 59--60, New York, NY, USA, 2013. ACM.

\bibitem{VishnoiSpeed}
Nisheeth~K. Vishnoi.
\newblock The speed of evolution.
\newblock In {\em SODA}, 2015.

\bibitem{agt.bimatrix}
Bernhard von Stengel.
\newblock Equilibrium computation for two-player games in strategic and
  extensive form.
\newblock {\em Algorithmic Game Theory, eds. Nisan, Roughgarden, Tardos, and
  Vazirani}.

\bibitem{wikiintertia}
Wikipedia.
\newblock https://en.wikipedia.org/wiki/sylvester's\_law\_of\_inertia.

\end{thebibliography}

\newpage
\noindent \Large\textbf{Appendix}
\normalsize
\appendix

\section{Missing Proofs of Section \ref{sec.char}}\label{asec.char}
\subsection{Equations of Jacobian}
Since $f$ is defined on $n$ variables while its domain is $\Delta_n$ which is of $n-1$ dimension, we consider a projected Jacobian by replacing strategy $t$ with $x_t>0$ by $1-\sum_{i\neq t}x_i$ in $f$.

\begin{align*}
J_{ii}^{\xx} &= \frac{(A\xx)_{i}}{\xx^TA\xx}+x_{i}\frac{(A_{ii}-A_{it})(\xx^TA\xx) - 2(A\xx)_{i}((A\xx)_{i}-(A\xx)_{t})  }{(\xx^TA\xx)^2}\\
J_{ij}^{\xx} &= x_{i}\frac{(A_{ij}-A_{it})(\xx^TA\xx) - 2(A\xx)_{i}((A\xx)_{j}-(A\xx)_{t})  }{(\xx^TA\xx)^2}
\end{align*}

\noindent
If $\xx$ is a fixed-point of $f$, then using Remark \ref{lem.fpchar} the above simplifies to,

\begin{align*}
J_{ii}^{\xx} &= 1+x_{i}\frac{(A_{ii}-A_{it})}{\xx^TA\xx} \textrm{ if }x_{i}>0\ \  \textrm{ and }\ \  J_{ii}^{\xx} =
\frac{(A\xx)_{i}}{\xx^TA\xx}\textrm{ if }x_{i}=0\\
J_{ij}^{\xx} &= x_{i}\frac{A_{ij}-A_{it}   }{\xx^TA\xx} \textrm{ if
}x_{i},x_{j}>0\ \  \textrm{ and }\ \  J^{\xx}_{ij} = 0\textrm{ if }x_{i}=0
\end{align*}
\subsection{Proof of Theorem \ref{generic_zero}}\label{app.conv}
The proof of the following theorem follows the proofs that can be found in \cite{Soda15a}
To prove Theorem \ref{generic_zero}, we will make use of the following important theorem in dynamical systems.

\begin{theorem}\label{athm.manifold}(Center and Stable Manifolds, p. 65 of \cite{shub})
Let $\textbf{p}$ be a fixed point for the $C^r$ local diffeomorphism $h: U \to \mathbb{R}^n$ where $U \subset \mathbb{R}^n$ is an open
(full-dimensional) neighborhood of $\pp$ in $\mathbb{R}^n$ and $r \geq 1$. Let $E^s \oplus E^c \oplus E^u$ be the invariant splitting
of $\mathbb{R}^n$ into generalized eigenspaces of $Dh(\textbf{p})$\footnote{Jacobian of $h$ evaluated at $\pp$} corresponding to
eigenvalues of absolute value less than one, equal to one, and greater than one. To the $Dh(\textbf{p})$ invariant subspace $E^s\oplus
E^c$ there is an associated local $h$ invariant $C^r$ embedded disc $W^{sc}_{loc}$ of dimension $dim(E^s \oplus E^c)$, and 
ball $B$ around $\pp$ such that:
\begin{equation} h(W^{sc}_{loc}) \cap B \subset W^{sc}_{loc}.\textrm{  If } h^n(\textbf{x}) \in B \textrm{ for all }n \geq 0,
\textrm{ then }\textbf{x} \in W^{sc}_{loc}
\end{equation}
\end{theorem}

To use the theorem above we need to project the map of the dynamics \ref{eq.f} to a lower dimensional space.  We consider the (diffeomorphism) function
$g$ that is a projection of the points $\xx \in \mathbb{R}^{n}$ to $\mathbb{R}^{n-1}$ by excluding a specific
(the "first") variable. We denote this projection of $\Delta_{n}$ by $g(\Delta_{n})$, i.e., $\xx \to_ g (\xx')$ where $\xx'=(x_2,\dots,x_n)$. Further, we define the fixed-point dependent projection function $z_{\rr}$ where we remove one variable $x_{t}$ so that $\rr_{t}>0$ (like function $g$ but the removed strategy must be chosen with positive probability at $\rr$).
\\\\Let $f$ be the map of dynamical system (\ref{eq.f}). For a linearly unstable fixed point $\textbf{r}$ we consider the function
$\psi_{\textbf{r}}(\vv) = z_{\textbf{r}}\circ f \circ z_{\textbf{r}}^{-1}(\vv)$ which is $C^1$ local diffeomorphism (due to the
point-wise convergence of $f$ \cite{akin} we know that the rule of the dynamical system is a diffeomorphism),
with $\vv \in R^{n-1}$.
Let $B_{\rr}$ be the (open) ball that is derived from Theorem \ref{athm.manifold} and we consider the union of these balls
(transformed in $\mathbb{R}^{n-1}$) $$A = \cup _{\textbf{r}}A_{\textbf{r}}$$ where $A_{\textbf{r}} =
g(z_{\textbf{r}}^{-1}(B_{\textbf{r}}))$ ($z^{-1}_{\textbf{r}}$ "returns" the set $B_{\textbf{r}}$ back to
$\mathbb{R}^{n}$). Set $A_{\textbf{r}}$ is an open subset of $\mathbb{R}^{n-1}$ (by continuity of $z_{\textbf{r}}$).  Taking
advantage of separability of $\mathbb{R}^{n-1}$ we have the following theorem.

\begin{theorem} (Lindel\H{o}f's lemma) For every open cover there is a countable subcover.
\end{theorem}

Therefore due to the above theorem, we can find a countable subcover for $A$, i.e.,
there exists fixed-points $\rr_1,\rr_2,\dots$ such that $A = \cup _{m=1}^{\infty}A_{\textbf{r}_{m}}$.
\\\\ For a $t \in \mathbb{N}$ let $\psi_{t,\textbf{r}}(\vv)$ the point after $t$ iteration of dynamics (\ref{eq.f}), starting
with $\vv$, under projection $z_{\rr}$, i.e., $\psi_{t,\textbf{r}}(\vv) = z_{\textbf{r}}\circ f^{t} \circ
z_{\textbf{r}}^{-1}(\vv)$.
If point $\vv \in int \; g(\Delta_{n})$ (which corresponds to $g^{-1}(\vv)$ in our original $\Delta_{n}$) has a linearly unstable fixed point
as a limit, there must exist a $t_{0}$ and $m$ so that $\psi_{t,\textbf{r}_{m}} \circ z_{\textbf{r}_{m}} \circ
g^{-1}(\vv) \in B_{\textbf{r}_{m}}$ for all $t \geq t_{0}$ (we have point-wise convergence from \cite{akin}) and therefore
again from Theorem \ref{athm.manifold} we get that
$\psi_{t_{0},\textbf{r}_{m}} \circ z_{\textbf{r}_{m}} \circ g^{-1}(\vv) \in W_{loc}^{sc}(\rr_m)$,
hence $\vv \in g\circ z^{-1}_{\textbf{r}_{m}} \circ
\psi^{-1}_{t_{0},\textbf{r}_{m}}(W_{loc}^{sc}(\rr_m))$.\\\\ Hence the set of points in $int \;
g(\Delta_{n})$ whose $\omega$-limit has a linearly unstable equilibrium is a subset of
\begin{equation}
C=  \cup_{m=1}^{\infty} \cup_{t=1}^{\infty} g\circ z^{-1}_{\textbf{r}_{m}} \circ
\psi^{-1}_{t,\textbf{r}_{m}}(W_{loc}^{sc}(\rr_m))
\end{equation}

Since $\rr_m$ is linearly unstable, it has $dim(E^u)\ge 1$, and therefore dimension of $W_{loc}^{sc}(\rr_m)$ is at most $n-2$. Thus,
the manifold $W_{loc}^{sc}(\rr_m)$ has Lebesgue measure zero in $\mathbb{R}^{n-1}$. Finally since $g\circ z^{-1}_{\rr_{m}} \circ
\psi^{-1}_{t,\rr_{m}} : \mathbb{R}^{n-1} \to \mathbb{R}^{n-1}$ is continuously differentiable, $\psi_{t, \textbf{r}_{m}}$ is $C^1$ and
locally Lipschitz (see \cite{perko} p.71). Therefore using Lemma \ref{lips} below it preserves the null-sets, and thereby we get that
$C$ is a countable union of measure zero sets, i.e., is measure zero as well, and Theorem \ref{generic_zero} follows.

\begin{lemma}\label{lips} Appendix of \cite{Soda15a}. Let $g: \mathbb{R}^m \to \mathbb{R}^m$ be a locally  Lipschitz   function, then
$g$ is null-set preserving, i.e., for $E \subset \mathbb{R}^m$ if $E$ has measure zero then $g(E)$ has also measure zero.  \end{lemma}
\subsection{Proof of Lemma \ref{lem.t}}
\begin{proof} It suffices to assume that $\xx$ is fully mixed. Let $\zz$ be any vector with $\sum_i z_i=0$ and define the vector $\ww = (z_1-z_n,...,z_{n-1}-z_n)$ with support $n-1$. It is clear that $\zz^TA\zz \leq 0$ iff $\ww^TT(A,\xx)\ww \leq 0$. So if $\zz^TA\zz \leq 0$ for all $\zz$ with $\sum_i z_i=0$ then $T(A,\xx)$ is negative semidefinite. If there exists a $\zz$ with with $\sum_i z_i=0$ s.t $\zz^TA\zz > 0$  then $\ww^TT(A,\xx)\ww >0$, so $T(A,\xx)$ is not negative semidefinite.    
\end{proof}
\subsection{Proof of Lemma \ref{lem.sfp2ns}}
\begin{proof} There is a similar proof in \cite{lyubich} for a modified claim. Here we connect the two. First of
all, observe that if $(\rr,\rr)$ is not Nash equilibrium for the $(A,A)$ game then there exists a $j$ such that $r_{j}=0$
and $(A\rr)_{j} > \rr^TA\rr$. But $\frac{(A\rr)_{j} }{\rr^TA\rr}>1$ is an eigenvalue of $J^{\rr}$.

Additionally, since $\rr$ is stable, using Theorem \ref{maximum} we have that $\rr$ is a local maximum of $\pi(\xx)=\xx^T
A\xx$, say in a neighborhood $\|\xx-\rr\|< \delta$. Let $\yy$ be a vector with support subset of the support of $\rr$ such
that $\sum_{i}y_{i}=0$.  Firstly we rescale it w.l.o.g so that $\|\yy'\|< \delta$ and by setting $\zz = \yy'+\rr$ we have
that $$\zz^TA\zz = \yy'^TA\yy' + \rr^TA\rr + 2\yy'^TA\rr$$ But $\yy'^TA\rr=0$ since $(A\rr)_{i}=(A\rr)_{j}$ for all $i,j$
s.t $r_{i},r_{j}>0$, $\sum_{i}y'_{i}=0$ , $\yy'$ has support subset of the support of $\rr$. Therefore
$\yy'^TA\yy'+\rr^TA\rr = \zz^TA\zz \leq \rr^TA\rr$, thus $\yy'^TA\yy' \leq 0$.
Hence proved using Lemma \ref{lem.t}.
\end{proof}
\subsection{Proof of Lemma \ref{lem.sns2afp}}
\begin{proof} The proof can be found in in \cite{lyubich}. The sufficient conditions for a fixed point to be asymptotically stable in the proof are exactly the assumptions for a fixed-point to be strict Nash stable. 
\end{proof}
\subsection{Proof of Lemma \ref{lem.nseqls}} 
\begin{proof} Let $t$ be the removed strategy (variable $x_{t}$) to create $J^{\rr}$ (with $r_{t}>0$). For every $i$ such that $r_{i}=0$ we have that $J^{\rr}_{ii} = \frac{(A\rr)_{i}}{\rr^TA\rr }$ and $ J^{\rr}_{ij} =0$ for all $j \neq i$. Hence the corresponding eigenvalues of $J^{\rr}$ of the rows $i$ that do not belong in the support of $\rr$ ($i$-th row has all zeros except the diagonal entry $J^{\rr}_{ii} = \frac{(A\rr)_{i}}{\rr^TA\rr }$) are $\frac{(A\rr)_{i}}{\rr^TA\rr }>0$ which are less than or equal to 1 iff $(\rr,\rr)$ is a NE of the game $(A,A)$.  

Let $\mathbb{J}^{\rr}$ be the submatrix of $J^{\rr}$ by removing all columns/rows $j \notin SP(\rr)$. Let $A'$ be the submatrix of $A$ by removing all columns/rows $j \notin SP(\rr)$. It suffices to prove that $T(A,\rr)$ is negative semi-definite iff $\mathbb{J}^\rr$ has eigenvalues with absolute value at most 1. 

Let $k = |SP(\rr)|$ and the $k \times k$ matrix $L$ with $L_{ij}= r_{i} \frac{A_{ij}}{\rr^TA\rr}$ and $i,j \in SP(\rr)$. Observe that $L$ is stochastic and also symmetrizes to $L'_{ij}= \sqrt{r_{i}r_{j}} \frac{A_{ij}}{\rr^TA\rr}$, i.e $L,L'$ have the same eigenvalues. Therefore $L$ has an eigenvalue 1 and the rest eigenvalues are real between $(-1,1)$ and also $A'$ has eigenvalues with the same signs as $L'$. 

Finally, we show that $det(L - \lambda I_{k}) = (1-\lambda) \times det(\mathbb{J}'^{\rr}  - \lambda I_{k-1})$ with $\mathbb{J}'^{\rr} = \mathbb{J}^{\rr} - I_{k-1}$, namely $L$ has the same eigenvalues as $\mathbb{J}'^{\rr}$ plus eigenvalue $1$. It is true for a square matrix that by adding a multiple of a row/column to another row/column, the determinant stays invariant. We consider $ L - \lambda I_{k}$ and we do the following: We subtract the $t$-th column  from every other column and on the resulting
matrix, we add every row to the $t$-th row. The resulting matrix $R$  has the property that $det(R) =
(1-\lambda)\times det(\mathbb{J}'^{\rr} - \lambda I_{k-1})$ and also $det(R) = det(L - \lambda I_{k})$. 

From above we get that if $\mathbb{J}^{\rr}$ has eigenvalues with absolute value at most 1, then $\mathbb{J}^{\rr} - I_{k-1}$ has eigenvalues in $[-2,0]$ (we know that are real from the fact that $L$ is symmetrizes to $L'$), hence $L'$ has eigenvalue 1 and the rest eigenvalues are in $[-2,0]$ (since $L$ is stochastic, the rest eigenvalues lie in $(-1,0]$ ). Therefore $A'$ has positive inertia 1 (see Sylvester's law of inertia) and the one direction follows. For the converse, $T(A,\rr)$ being negative semi-definite implies $A'$ has positive inertia 1, thus $L'$ and so $L$ have one eigenvalue positive (which is 1) and the rest non-positive (lie in $(-1,0]$ since $L$ is stochastic). Thus $\mathbb{J}^{\rr}- I_{k-1}$ has eigenvalues in $(-1,0]$ and therefore $\mathbb{J}^{\rr}$ has eigenvalues in $(0,1]$ (i.e with absolute value at most 1).
\end{proof}

\subsection{Proof of Lemma \ref{lem.inv}}
\begin{proof}
For equivalence of (strict) Nash stable points, the set of (strict) symmetric NE are same for games $(A,A)$ and $(B,B)$, and
matrix $T(A,\xx)=T(B,\xx), \forall \xx \in \Delta_n$.
\end{proof}

\section{Missing Proofs of Section \ref{sec.hard} (Hardness)}\label{asec.hard}
\subsection{Proof of Lemma \ref{lem.cl2sns}}
\begin{proof}
Let vertex set $C\subset V$ forms a clique of size $k$ in graph $G$. Construct a maximal clique containing $C$, by adding
vertices that are connected to all the vertices in the current clique. Let the corresponding vertex set be $S\subset V$
($C\subset S$), and let $m = |S| \ge k$. Wlog assume that $S=\{v_1,\dots,v_{m}\}$. Now we construct a strategy profile $\pp\in
\Delta_{2n}$ and show that it is a strict Nash stable of game $(A,A)$. 

\[ p_i = \left\{
  \begin{array}{l l }
	\frac{1}{m} & \quad \text{$1\le i \le m$}\\
	0 & \quad \text{$m+1\le i \le 2n$}
  \end{array} \right.\]

\begin{claim}
$\pp$ is a strict SNE of game $(A,A)$.
\end{claim}
\begin{proof}
To prove the claim we need to show that $(A\pp)_i > (A\pp)_j,\ \forall i \in [m], \forall j \notin [m]$, and
$(A\pp)_i=(A\pp)_j,\ \forall i,j \in [m]$. Since $S$ forms a clique in graph $G$, and by construction of $A$,  
the payoff from $i^{th}$ pure strategy against $\pp$ is
\[
(A\pp)_i = \sum_{r\le m} A_{ir} p_r = 
\left\{\begin{array}{ll}
\sum_{r\le m} \frac{1}{m} = \frac{m-1}{m},  & \forall i \in [m] \\ 
\sum_{r\le m, E_{ir}=1} \frac{1}{m}\text{-}\sum_{r\le m, E_{ir}=0} h < \frac{m-1}{m}  & \forall m < i \le n \ \ (\because\ \exists r \le m,
E_{ir}=0) \\ 
\frac{k-1}{m}-\epsilon(1-\frac{1}{m}) < \frac{k-1}{m} \le \frac{m-1}{m} & \forall n < i \le 2n \ \ (\because m \ge k \mbox{ and } k < n-1)
\end{array} \right. 
\]
Thus the claim follows.
\end{proof}

Next consider the corresponding transformed matrix $B=T(A,[m])$ as defined in (\ref{eq.t}). Since $A_{ij}=1 \forall i,j \in
[m], i\neq j$ and $A_{ii}=0,\ \forall i \in [m]$, we have 
\[
\begin{array}{ll}
\forall i,j<m,\ B_{ij} = A_{ij}+A_{mm}-A_{im}-A_{mj} & = -1 \mbox{ if } i \neq j \\
 = -2 \mbox{ if } i=j
\end{array}
\]
\begin{claim}
$B$ is negative definite.
\end{claim}
\begin{proof}
It is easy to check that $B$ has all strictly negative eigenvalues. $\ww^1=\ones_{m-1}$ is an eigenvector with eigenvalue $-m$, and
$\forall 1<i<m$, vector $\ww^i$, where $w^i_1=1$ and $w^i_i=-1$, is an eigenvector with eigenvalue $-1$. Further,
$\ww^1,\dots,\ww^{m-1}$ are linearly independent.
\end{proof}

Thus by Definition \ref{def.ns} $\pp$ is a strict Nash stable for game $(A,A)$.
\end{proof}

\subsection{Proof of Lemma \ref{lem.dom}}
\begin{proof} A negative semi-definite matrix has the property that all the diagonal elements are non-positive. Observe that
from definition of $T(A,\rr)$, we can choose any strategy to be removed that is in $SP(\rr)$, hence we choose $i$ and we
look at entry $B_{jj} = A_{ii} + A_{jj} - 2 A_{ij}$ with $j \in SP(\rr), j\neq i$ which must be non-positive since
$T(A,\rr)$ is negative semi-definite. Hence $A_{ii} \leq A_{ii}+A_{jj} \leq 2 A_{ij}$. Finally, if $A_{ii},A_{jj}$ are both
dominating then $A_{ii}+A_{jj} > A_{ij}+ A_{ji}= 2A_{ij}$ which is contradiction since $A_{ii} + A_{jj} - 2 A_{ij} \leq 0$.
\end{proof}

\subsection{Proof of Lemma \ref{lem.sne2cl}}
\begin{proof}
Let's define $\ssp(\pp)=\{i \ |\  p_i>\frac{1}{n^2}\}$. We first show $|\ssp(\pp)|\geq k$.

\begin{claim}
$|\ssp(\pp)|\geq k$
\end{claim}
\begin{proof}
Note that $\sum_{i\in SP(\pp)\setminus \ssp(\pp)} p_i \leq n\frac{1}{n^2} \leq \frac{1}{n}$. Therefore, $\sum_{i\in \ssp(\pp)} p_i
\geq 1-\frac{1}{n}$. Suppose $|\ssp(\pp)|<k$ by contradiction. Then $\exists r\in \ssp(\pp)$ such that
$p_r\ge\frac{1-\frac{1}{n}}{k-1}=\frac{n-1}{n(k-1)}$. Now consider the payoff from strategy $n+r$, which is $$(A\pp)_{n+r} =
(k-1)p_r-\epsilon(1-p_r) \geq 1-\frac{1}{n} - \epsilon$$ On the other hand we have  
$$ (A\pp)_r \le 1-p_r \le
1-\frac{n-1}{n(k-1)}$$ Therefore, $$(A\pp)_{n+r}-(A\pp)_r \geq 1-\frac{1}{n}-\epsilon
-1+\frac{n-1}{n(k-1)}=\frac{n-k}{n(k-1)} - \epsilon\geq \frac{n-k}{n(k-1)} - \frac{1}{10n^3}>0 $$
A contradiction to $\pp$ being symmetric $NE$.
\end{proof}

Let's define $S=\{v_i\ | \ i \in \ssp(\pp)\}$. In order to prove the lemma, it suffices to show the vertex set $S$ forms a clique in the graph $G$ since $|S|=|\ssp(\pp)|\geq k$.

\begin{claim}
The vertex set $S$ forms a clique in the graph $G$.
\end{claim} 
\begin{proof}
It suffices to show $\forall i,j \in \ssp(\pp)$ where $i\not = j$ we have $A_{ij}=1$. Suppose not then $\exists i',j'\in \ssp(\pp)$ s.t. $i'\not =j'$ and $A_{i'j'}\not =1$. We get $A_{i'j'}=-h$ by definition of $A$. Therefore, $(A\pp)_{j'}\leq -hp_{i'} +1\le -1$ because $p_{i'}\geq \frac{1}{n^2}$ and $h\geq 2n^2$. On the other hand, we have $\forall i\not \in [n]$, $(A\pp)_i \geq -\epsilon > -1$ by definition of $A$ so we get a contradiction to $\pp$ being symmetric NE.
\end{proof}

The proof is completed.
\end{proof}

\subsection{Proof of lemma \ref{lem.ns2cl}}
\begin{proof}
Let $\pp$ be a Nash stable strategy of game $(A,A)$, then by definition $\pp$ is a SNE and matrix $B=T(A,\pp)$ is negative
semi-definite. The latter implies $SP(\pp) \subset [n]$ using Lemma \ref{lem.dom}, since for $i \notin[n]$
$A_{ii}=h>2k>2A_{ij}, \forall j\neq i$. Applying Lemma \ref{lem.sne2cl} with this fact together with the $\pp$ being
an SNE and $|SP(\pp)| > 1$ implies $G$ has a clique of size $k$.
\end{proof}

\section{Missing Details of Section \ref{sec.hard} (Hardness)}

\subsection{Hardness when single dominating diagonal}\label{asec.ddhard}
A symmetric matrix, when picked uniformly at random, has in expectation exactly one row with dominating diagonal (see Remark
\ref{rem}). One could ask does the problem become easier for this typical case. We answer negatively by extending all the
NP-hardness results of Theorem \ref{thm.mainHard} to this case as well, where matrix $A$ has exactly one
row whose diagonal entry dominates all other entries of the row, {\em i.e.,} $\exists i:\ A_{ii}>A_{ij},\ \forall j \neq i$.

Consider the following modification of matrix $A$ from (\ref{eq.red1}), where we add an extra
row and column. Matrix $M$ is of dimension $(2n+1) \times (2n+1)$, described pictorially in Figure \ref{fig.m}.
Recall that $h>2n^2+5$ and $k$ is the given integer.

\begin{equation}\label{eq.m}
\begin{array}{ll}
M_{ij}= A_{ij} & \mbox{ if } i,j \le 2n \\
M_{(2n+1)i}=M_{i(2n+1)} = 0 & \mbox{ if } i\le n \\
M_{(2n+1)i}=M_{i(2n+1)} = h+\epsilon & \mbox{ if } n<i\le 2n, \mbox{ where $0<\epsilon<1$ } \\
M_{(2n+1)(2n+1)}=3h
\end{array}
\end{equation}

\begin{figure}[!htb]
\centering
\includegraphics[width=.5\linewidth]{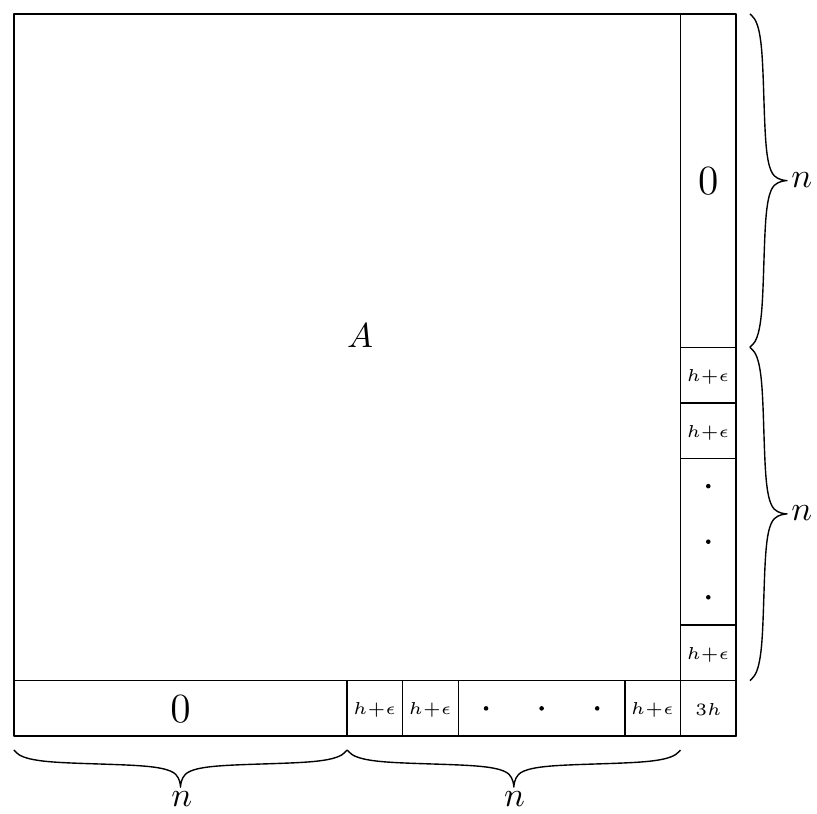}
\caption{Matrix $M$ as defined in (\ref{eq.m})} 
\label{fig.m}
\end{figure}

Clearly $M$ has exactly one row/column with dominating diagonal, namely $(2n+1)$. The strategy constructed in Lemma
\ref{lem.cl2sns} is still strict Nash stable in game $(M,M)$. This is because their support is a subset of $[n]$,
implying the extra strategy giving zero payoff which is strictly less than the expected payoff. Thus, we get the following lemma.

\begin{lemma}\label{lem.1cl2sns}
If graph $G$ has a clique of size $k$, then game $(M,M)$ has a mixed strategy that is strict Nash stable where $A$ is from
(\ref{eq.red1}). 
\end{lemma}

Next we show the converse. Nash stable strategies are super set of other three notion of stability (Theorem \ref{thm.char}), so it
suffices to map Nash stable to a $k$-clique. Further, if $\pp$ is Nash stable then $T(M,\pp)$ is negative semi-definite (by
definition). Using this property together with the lemmas from previous two section, we show the next lemma. 

\begin{lemma}\label{lem.1ns2cl}
Graph $G$ has a clique of size $k$, if there is a mixed Nash stable strategy in $(M,M)$.
\end{lemma}
\begin{proof}
Let $\qq$ be the Nash stable strategy, then it is a symmetric NE of game $(M,M)$, $T(M,\qq)$ is
negative semi-definite, and $|SP(\qq)|>1$. Using Lemma \ref{lem.dom} we have $SP(\qq)\subseteq [n]$, as for $i=2n+1,\
M_{ii}=3h>2M_{ij},\ \forall j\neq i$ implying $2n+1\notin SP(\qq)$, and $\forall n<i\le 2n,\ M_{ii}=h>2k >2A_{ij},\ \forall j \in
SP(\qq)$. Thus for $2n$-dimensional vector $\pp$, where $p_i =q_i, i\le 2n$, we have $T(A,\pp)=T(M,\qq)$ and $\pp$ is a symmetric NE
of game $(A,A)$. Thus, stability property of $\qq$ on matrix $M$ carries forward to corresponding stability of $\pp$ on matrix $A$.
Thus the lemma follows using Lemma \ref{lem.ns2cl}. 
\end{proof}

The next theorem follows using Theorem \ref{thm.char}, Lemmas \ref{lem.inv}, \ref{lem.1cl2sns}, and
\ref{lem.1ns2cl}. Containment in NP follows using the Definition \ref{def.ns}.

\begin{theorem}\label{thm.1ns}
Given a symmetric matrix $M$ such that exactly one row/column in $M$ has a dominating diagonal, 
\begin{itemize}
\item it is NP-complete to check if game $(M,M)$ has a mixed Nash stable (or linearly stable) strategy.
\item it is NP-complete to check if game $(M,M)$ has a mixed strict Nash stable strategy.
\item it is NP-hard to check if dynamics (\ref{eq.f}) applied on $M$ has a mixed stable fixed-point.
\item it is NP-hard to check if dynamics (\ref{eq.f}) applied on $M$ has a mixed asymptotically stable fixed-point.
\end{itemize}
even if $M$ is a assumed to be positive.
\end{theorem}

Strict positivity of the matrix in the above theorem follows using the fact that Nash stable and strict Nash stable strategies do not
change when a constant is added to the matrix (Lemma \ref{lem.inv}).

\subsection{Hardness for subset}\label{asec.subset} 
Another natural question to ask is whether a particular allele is going to survive with positive
probability in the limit, for a given fitness matrix. We show that this may not be easy either, by proving hardness for
checking if there exists a stable strategy $\pp$ such that $i\in SP(\pp)$ for a given $i$. In general, given a subset $S$
of pure strategies it is hard to check if $\exists$ a stable profile $\pp$ such that $S$ is a subset of $SP(\pp)$.

\begin{theorem}\label{thm.ss}
Given a $d\times d$ symmetric matrix $M$ and a subset $S\subset [d]$,
\begin{itemize}
\item it is NP-complete to check if game $(M,M)$ has a Nash stable (or linearly stable) strategy $\pp$ s.t. $S\subset
SP(\pp)$.
\item it is NP-complete to check if game $(M,M)$ has a strict Nash stable strategy $\pp$ s.t. $S\subset SP(\pp)$.
\item it is NP-hard to check if dynamics (\ref{eq.f}) applied on $M$ has a stable fixed-point $\pp$ s.t. $S\subset
SP(\pp)$.
\item it is NP-hard to check if dynamics (\ref{eq.f}) applied on $M$ has a asymptotically stable fixed-point $\pp$ s.t.
$S\subset SP(\pp)$.
\end{itemize}
even if $|S|=1$, or if $M$ is a assumed to be positive or with exactly one row with dominating diagonal.
\end{theorem}
\begin{proof}
The reduction is again from $k$-clique. Our constructions of (\ref{eq.red1}) and (\ref{eq.m}) works as
is, and the target set is $S=\{n\}$. 

Recall that vertex $v_n \in V$ is connected to every other vertex in $G$, and therefore is part of every maximal
clique. Thus the construction of strategy $\pp$ in Lemmas \ref{lem.cl2sns} and \ref{lem.1cl2sns} will have
$p_n>0$, and therefore if $k$-clique exist then $S\subset SP(\pp)$. 

For the converse consider a Nash stable strategy $\pp$ with $\{n\}\subset SP(\pp)$. By definition it is a symmetric NE, and
therefore $SP(\pp)\neq\{n\}$ as in all cases row $n$ has dominated diagonal, i.e., $A_{nn}=0 < k-\delta=A_{n,2n}$. Thus,
$\pp$ is a mixed profile, and then by applying Lemmas \ref{lem.ns2cl} and
\ref{lem.1ns2cl}, for the respective cases we get that graph $G$ has a $k$-clique. Thus proof follows using
Theorem \ref{thm.char} and property of Lemma \ref{lem.inv}.
\end{proof}

\subsection{Diversity and hardness}\label{asec.divNhard}
Finally we state the hardness result in terms of survival of phenotypic diversity in the limiting population of diploid
organism with single locus. For this case, as we discussed before, the evolutionary process has been studied extensively
\cite{Nagylaki1,akin,lyubich}, and that it is governed by dynamics $f$ of (\ref{eq.f}) has been established.
Here $A$ is a symmetric fitness matrix; $A_{ij}$ is the fitness of an organism with alleles $i$ and $j$ in the locus of two
chromosomes. Thus, for a given $A$ the question of deciding ``If phenotypic diversity will survive with positive
probability?'' translates to ``If dynamics $f$ converges to a {\em mixed} fixed-point with positive probability?''. We wish
to show NP-hardness for this question. 

Theorem \ref{generic_zero} establishes that all, except for zero-measure, of starting distributions $f$ converges to
linearly stable fixed-points. From this we can conclude that ``Yes'' answer to the above question implies 
existence of a {\em mixed} linearly stable fixed-point. However the converse may not hold.  In other
words, ``No'' answer does not imply {\em non-existence} of {\em mixed} linearly stable fixed-point. 
Although, in that case we can conclude non-existence of {\em mixed} strict Nash stable strategy (Theorem \ref{thm.char}). Thus, none of the
above reductions seem to directly give NP-hardness for our question.

At this point, the fact that same reduction (of Section \ref{sec.SHard}) gives NP-hardness for all four notions of stablity, and in
particular for strict Nash stable as well as linearly stable (Nash stable) fixed-points come to our rescue. In particular, for the
matrix $A$ of (\ref{eq.red1}) non-existence of {\em mixed} limit point in $\CL$ (points where $f$ coverges with positive probability)
implies non-existence of strict Nash stable strategy, which in turn imply non-existence of {\em mixed} linearly stable
fixed-point (Theorem \ref{thm.char}). If not, then graph $G$ will have $k$-clique (Lemma \ref{lem.ns2cl} and Theorem \ref{thm.char}),
which in turn implies existence of a mixed strict Nash stable strategy (Lemma \ref{lem.cl2sns}). Therefore, we can conclude that {\em
mixed} linearly stable fixed-point exist if and only if $f$ converges to a {\em mixed} fixed-point with positive probability. and thus
the next theorem follows.

\begin{theorem}\label{athm.main}
Given a fitness matrix $A$ for a diploid organism with single locus, it is NP-hard to decide if, under evolution, diversity will
survive (by converging to a specific mixed equilibrium with positive probability) when starting allele frequencies are picked iid from uniform distribution. Also, deciding if a given
allele will survive is NP-hard.  
\end{theorem}

\begin{remark}\label{rem.cg}
As noted in Section \ref{sec.prelgame}, coordination games are very special and they always have a pure Nash equilibrium which is easy
to find; NE computation in general game is PPAD-complete \cite{DGP}.  Thus, it is natural to wonder if decision versions on
coordination games are also easy to answer.

In the process of obtaining the above hardness results, we stumbled upon NP-hardness for checking if a symmetric coordination game has
a NE (not necessarily symmetric) where each player randomizes among at least $k$ strategies. Again the reduction is from $k$-clique.
Thus, it seems highly probable that other decision version on (symmetric) coordination games are also NP-complete.
\end{remark}

\section{Survival of Diversity}
\label{asec.survive}

Given a fitness (positive symmetric) matrix $A$. We analyze how diverse the population will be in the limit under
evolutionary dynamics (governed by $f$ (\ref{eq.f})). If the limit point $\xx$ is not pure, i.e., $|SP(\xx)|>1$ then at least
two alleles survive among the population, and we say the population is diverse and not monomorphic in the limit.

We characterize two extreme cases of fitness matrix for the survival of diversity, namely where diversity always
survives and where diversity disappears regardless of the starting population. Using this characterization we analyze the
chances of survival of diversity when fitness matrix and starting populations are picked uniformly at random. Since every limit with postive measure of region of attraction is Nash stable (Theorems \ref{generic_zero} and \ref{thm.char}), there has to be at least one mixed Nash stable strategy for the diversity to survive. We know that given a symmetric $A$ there always exists a Nash stable strategy, however it does not have to be mixed. 

\begin{definition} Diagonal entry $A_{ii}$ is called {\em dominating} if and only if $A_{ii} > A_{ij}$ for all $j \neq i$. And it
is called {\em dominated} if and only if $\exists j,\textrm{ such that } A_{ij}>A_{ii}$.
\end{definition}

Next lemma characterizes instances that lack mixed Nash stable. 

\begin{lemma}\label{only_pure} If all diagonal entries of $A$ are dominating then there are no {\em mixed} linearly stable fixed
points.  \end{lemma}

The next theorem follows using Theorem \ref{generic_zero} and Lemma \ref{only_pure}. Informally, it states that if every diagonal
entry of $A$ is dominating then almost surely the dynamics converge to pure fixed points.

\begin{theorem}\label{thm.die} If every diagonal entry of $A$ is dominating then the set of initial conditions in $\Delta_{n}$
so that the dynamics \ref{eq.f} converges to mixed fixed points has measure zero, i.e diversity dies almost surely.
\end{theorem}

Additionally, we show that a pure fixed point that chooses strategy $t$, where $A_{tt}$ is dominated, will be linearly
unstable.

\begin{lemma}\label{only_mixed} 
Let $\rr \in \Delta_n$ with $r_t=1$. 
If $A_{tt}$ is dominated, then $\rr$ is linearly unstable. 
\end{lemma}

If there are no pure fixed points that are linearly stable, then all linearly stable fixed-points are mixed, and thus the
next theorem follows using Theorem \ref{generic_zero} and Lemma \ref{only_mixed}. 

\begin{theorem} If every diagonal of $A$ is dominated then the set of initial conditions in $\Delta_{n}$ so that the
dynamics \ref{eq.f} converge to pure fixed points has measure zero, i.e diversity survives almost surely.  \end{theorem}

The following lemma shows that when the fitness matrix is picked uniformly at random, there is a 
positive probability (bounded away from zero for all $n$) so that
every diagonal in $A$ is dominated. This essentially means that generically, diversity survives with positive
probability, bounded away from zero, for all $n$ where the randomness is taken with respect to both the payoff matrix and
initial conditions.

\begin{lemma}\label{lem.pp} Let each entry of $A$ be chosen iid from a continuous distribution. The probability that all 
diagonals of $A$ are dominated is bounded away from zero, i.e., at least $\frac{1}{3} - o(1)$.
\end{lemma}
\begin{proof}
Let $E_{i}$ be the event that $A_{ii}$ is \textit{dominating}. We get that:\\
\begin{equation}
\Pr[E_{i}] = \frac{1}{n}
\end{equation}
Also for $n \geq 6$ we have that 
\begin{equation}
\Pr[E_{i} \cap E_{j} \cap E_{k}] \leq \frac{1}{(n-3)^3}
\end{equation} for $i \neq j \neq k \neq i$. To prove this let $D_{i}$ correspond to the events $A_{ii} > A_{it}$ for all $t \neq i,j,k$ (in same way the definition of $D_{j},D_{k}$). Clearly $D_{i},D_{j},D_{k}$ are independent and thus $\Pr[D_{i} \cap D_{j} \cap D_{k}] = \frac{1}{(n-3)^3}$.  Since $E_{i} \cap E_{j} \cap E_{k} \subset D_{i} \cap D_{j} \cap D_{k}$ the inequality follows.\\\\
Finally by counting argument (count all the favor permutations) we get that $$\Pr[E_{i} \cap E_{j}] \geq \frac{2[\sum_{k=0}^{n-2}(2n-3-k)! \frac{(n-2)!}{(n-2-k)!}]}{(2n-1)!} =  \frac{2}{n(n-1)} \sum_{k=0}^{n-2} \prod_{i=0}^{k+1}\frac{n-i}{2n-i-1}$$ for $i \neq j$. For $l = o(n)$, for example $l = \log n$ and using the fact that $\frac{n-i}{2n-i-1}$ is decreasing with respect to $i$ we get that 
\begin{align*}
\sum_{k=0}^{n-2}\prod_{i=0}^{k+1}\frac{n-i}{2n-i-1} &\geq \sum_{k=0}^{l}\prod_{i=0}^{k+1}\frac{n-i}{2n-i-1} \\& 
\geq \sum_{k=0}^{l}\left (\frac{n - k-1}{2n - k -2}\right)^{k+2} 
\\&\geq \sum_{k=0}^{l}\left(\frac{n - l-1}{2n - l -2}\right)^{k+2} \\ &
= \left(\frac{n-l-1}{2n-l-2}\right)^2 \left(\frac{1 - (\frac{n-l-1}{2n-l-2})^{l+1}}{1 - \frac{n-l-1}{2n-l-2}}\right) = \frac{1}{2}-o(1)
\end{align*}
 Therefore (inclusion-exclusion) we have that 
$$\Pr[\cup E_{i}] \leq \sum_{i}\Pr[E_{i}] - \sum_{i<j} \Pr[E_{i} \cap E_{j}] + \sum_{i<j<k}\Pr[E_{i} \cap E_{j} \cap E_{k}]
$$ thus 
\begin{align*}
\Pr[\cap E_{i}^c] \geq \frac{1}{2} - o(1) -\frac{n(n-1)(n-2)}{6(n-3)^3} = \frac{1}{3} - o(1)
\end{align*} which is bounded away from zero.
\end{proof}

\begin{remark}\label{rem}
Observe that letting $X_{i}$ be the indicator random variable that $A_{ii}$ is dominating and $X = \sum_{i}X_{i}$ we get
that $E[X] = \sum_{i} E[X_{i}] = \sum_{i} \Pr[E_{i}] = n \times \frac{1}{n}=1$ so in expectation we will have one dominating
element. Also from the above proof of Lemma \ref{lem.pp} we get that $E[X^2] = \sum_{i}E[X_{i}] + 2\sum_{i <j}E[X_{i}X_{j}] = 1
+  n(n-1) \Pr[E_{i} \cap E_{j}] \approx 2 - o(1)$ (namely $Var[X] \approx 1 - o(1)$) so by Chebyshev's inequality  $\Pr[|X -
1|>k]$ is $O(\frac{1}{k^2})$.
\end{remark}

\section{Terms Used in Biology}\label{asec.bioTerms}
\subsection{Terms in Biology}
\label{termsinbio}
We provide brief non-technical definitions of a few biological terms that we use in this paper.

\noindent{\bf Gene.}
A unit that determines some characteristic of the organism, and passes traits to offsprings.
All organisms have genes corresponding to various biological traits, some of which are instantly visible, such as eye color or number
of limbs, and some of which are not, such as blood type.
\medskip

\noindent{\bf Allele.}
Allele is one of a number of alternative forms of the same gene, found at the same place on a chromosome,
Different alleles can result in different observable traits, such as different pigmentation.
\medskip

\vspace{-10pt}
\noindent{\bf Genotype.}
The genetic constitution of an individual organism.
\medskip

\noindent{\bf Phenotype.}
The set of observable characteristics of an individual resulting from the interaction of its genotype with the environment.
\medskip


\noindent{\bf Diploid.}
Diploid means having two copies of each chromosome. Almost all of the cells in the human body are diploid.
\medskip

\noindent{\bf Haploid.}
A cell or nucleus having a single set of unpaired chromosomes.
Our sex cells (sperm and eggs) are haploid cells that are produced by meiosis. When sex cells unite during fertilization, the haploid
cells become a diploid cell.

\subsection{Heterozygote Advantage (Overdominance)}\label{over} 
 Cases of heterozygote advantage have been demonstrated in several organisms. The first confirmation of heterozygote advantage was with a fruit fly, Drosophila melanogaster. Kalmus demonstrated  in a classic paper~\cite{Kalmus45} how polymorphism can persist in a population through heterozygote advantage.  In humans, sickle-cell anemia is a genetic disorder caused by the presence of two recessive alleles. Where malaria is common, carrying a single sickle-cell allele (trait) confers a selective advantage, i.e., being a heterozygote is advantageous. Specifically, humans with one of the two alleles of sickle-cell disease exhibit less severe symptoms when infected with malaria. Theorem \ref{thm.die} and lemma \ref{lem.pp} are related to that phenomenon.

\end{document}